\documentclass[12pt,onecolumn]{IEEEtran}
\usepackage{latexsym,amssymb,amsmath,graphicx,bbm}
\usepackage{color}
\usepackage{epsfig}

\newcommand{\openone}{\leavevmode\hbox{\small1\normalsize\kern-.33em1}}

\newtheorem{definition}{Definition}

\newtheorem{theorem}{Theorem}
\newtheorem{Theorem}{Theorem}
\newtheorem{corollary}{Corollary}
\newtheorem{proposition}{Proposition}
\newtheorem{Proposition}{Proposition}

\def\proof{\par{\it Proof}. \ignorespaces}
\newcommand{\proofP}[1]{\par{\it Proof #1}. \ignorespaces}

\newcommand{\dist}{d_{\text{min}}}
\newcommand{\dmin}{d_{\text{min}}}
\newcommand{\dmoy}{\Delta}

\newcommand{\BC}{{\mathcal B}}
\newcommand{\RBC}{R_b}

\newcommand{\db}{\bar{\lambda}}

\newcommand{\eqdef}{\mbox{$\stackrel{\scriptstyle{\rm def}}{=}$}}

\newcommand{\Ar}{{\mathcal A}}
\newcommand{\dA}{\Delta{\mathcal A}}

\newcommand{\graph}{graph of codewords of partial weight 2}

\begin{document}

\title{\Large{Designing a Good Low-Rate Sparse-Graph Code}}
\author{
\authorblockN{\footnotesize Iryna Andriyanova$*$ and Jean-Pierre Tillich$^{\dagger}$\\}
\authorblockA{\footnotesize $^*$ETIS group, ENSEA/UCP/CNRS-UMR8051, France, iryna.andriyanova@ensea.fr\\
$^{\dagger}$Equipe-Projet SECRET, INRIA Roquencourt, France,
jeanpierre.tillich@inria.fr}
}

\maketitle

\vspace{-2cm}
\begin{abstract}
This paper deals with the design of low-rate sparse-graph codes with linear minimum distance ($d_{min}$) in the blocklength.
First, we define a necessary condition that a  quite general family of 
graphical  codes has to satisfy in order to have linear $d_{min}$.
The condition 
generalizes results known for turbo codes \cite{Breil} and LDPC codes.     
Secondly, we try to justify the necessity to introduce degree-1 bits (transmitted or punctured) into the code structure, while designing an efficient low-rate code. 
As a final result of our investigation, we present a new ensemble of low-rate codes, 
designed under the necessary condition and having bits of degree 1. 
The asymptotic analysis of the ensemble shows that its iterative threshold is close to the Shannon limit. 
It also has linear $d_{min}$, a simple structure and enjoys a low decoding complexity and a fast convergence. 
\end{abstract}
 
\section{Introduction} \label{sec:intro}

Low rate codes play a crucial role in communication systems operating in the low signal-to-noise ratio regime, such as   
power-limited sensor networks, ultra-wideband communications schemes and code-spread CDMA systems. 
More recently, it has also been found out that  powerful low-rate codes  with a fast decoding algorithm can be used in the reconciliation phase of continuous-variable quantum key distribution protocols and allow to increase significantly the range of the  protocol \cite{LG09}.

Since the invention of turbo codes \cite{BGT93}, a lot of effort was put into designing sparse-graph codes for various applications. 
This is due to nice features of the iterative decoding algorithm which is used in such codes, namely its low decoding complexity and good performance.  
Although the design of good low-rate sparse-graph codes is of great interest, it is not straightforward. 
By a good low-rate code ensemble we mean an
ensemble with 
iterative threshold close to the channel capacity and a good minimum distance, 
which is necessary to obtain a low error floor.  
The problem lies  in the fact that, in order to design a low rate code with 
performance close to the channel capacity, it seems crucial to have a large fraction of variable nodes of degrees 1 and 2 in the code structure
But the presence of a large number of variable nodes of low degrees is not favorable for the minimum distance growth. 
It may become logarithmic or, even worse, constant. 
This phenomenon has been quantified in several papers, such as for instance in \cite{PisFek06,TilZem06,OTA07}.
A way to circumvent the problem is to introduce some structure in the bipartite graph of a low-rate ensemble, preventing the formation of low-weight codewords. 

Recently, some high-performance low-rate schemes have been proposed. 
A rate-$1/10$ multi-edge LDPC ensemble with the threshold -1.09 dB on the AWGN channel was presented in \cite{RU06:multiedge}.
This construction can be viewed as a serial concatenation of a (3,15) LDPC code and of an LT code and
it possesses a complex structure. 
Its minimum distance growth 
inherits the minimum distance property of the underlying (3,15) LDPC inner code, i.e. it is linear in blocklength. 
In \cite{DDJ:ara-lowrate}, authors introduced 
low-rate ARA-type LDPC codes of different rates in the range from $1/3$ to $1/10$. 
The proposed ensembles have iterative thresholds close to the channel capacity and a simpler structure, compared to the previous multi-edge ensemble, but their minimum distance grows only polynomially in the blocklength\footnote{more precisely, it is of order $O(n^{3/4})$, see \cite{OT08}}.   
Also, in \cite{LYPW06:zigzaghadamard}, authors presented a
parallel concatenation of Zigzag-Hadamard (ZH) codes. 
These codes are decoded in a turbo-like fashion, by using the fast Hadamard transform for small Hadamard component codes. 
This yields a rather low complexity decoding algorithm.
The concatenated ZH ensembles have rates down to $0.00105$.
As for their minimum distance, the reasoning from \cite{TZ2006} can be adapted to show 
that the minimum distance of such a construction is of order $n^{(M-1)/M}$, where $n$ is the blocklength and $M$ is the number of component ZH codes.
{ This case is treated in \cite{OT08}.}

In this work, we propose an alternative low-rate code structure, which enjoys a good minimum distance, a good iterative threshold and a low decoding complexity. 
Our approach avoids to fix a complex bipartite graph structure and enables to get a flexible irregular construction.
Hence, the degree distribution of this construction can be optimized by a simple one-dimensional optimization.
The procedure that we adapt is the following:
a) we first provide a necessary condition to ensure linear 
minimum distance, 
b) then we design a low-rate code ensemble which satisfies this
condition based on
a component code that enjoys
a low-complexity decoding algorithm.   
To fulfill the first point, we define a special graph, called the {\emph \graph}. 
This graph is derived 
from connections of variable nodes of degrees 1 and 2 and of low-weight codewords of component codes. 
In some sense, it is a generalization of the subgraph induced by degree-2 variable nodes for LDPC codes \cite{DRU06:weightdist} to any sparse-graph code ensemble.

Tail-biting Trellis LDPC (or TLDPC) codes have been introduced in \cite{ATC05, ATC06}. 
This family enjoys an iterative threshold close to the channel capacity, a linear minimum distance and a very low decoding complexity.  
Examples of TLDPC codes of rates $1/3$ and $1/2$ were presented in \cite{ATC05, ATC06}.  
In this paper, we utilize the framework of TLDPC codes to design a code ensemble of lower rate.
We propose a new TLDPC component code, having a very simple structure 
The proposed component code has an interesting feature, which makes the obtaining of linear minimum distance possible: the supports of its low-weight codewords are distributed among the code positions in such a way that the union of intersecting supports form \emph{disjoint} clusters. 
We will discuss this property in details later on in the paper. 
We also emphasize that our choice of the component code allows to have a large non-zero fraction of degree-1 
bits
in the code structure, while keeping the minimum distance grow linearly in the blocklength. 
The presence of degree-1 bits improves the performance of iterative decoding, it will be explained later on in the paper. 

To design a low-rate TLDPC ensemble both with 
 linear minimum distance and an iterative threshold close to the channel capacity, we put a constraint on the maximum allowed fraction of degree-2 variable nodes and optimize over
 the degree distribution of variable nodes by using EXIT charts. 
Moreover, in order to satisfy the necessary condition for linear minimum distance we have found, we propose a structured way to generate the permutation for edges connected to degree-2 variable nodes. 
There is no other constraint on the generation of the permutation for other edges in the bipartite graph, it is assumed  to be drawn uniformly at random.

The paper is organized as follows. In the Section \ref{sec:graph} the \graph\ and a necessary condition for linear minimum distance are provided.
Section \ref{sec:degree_one} gives an insight why it is important to put degree-1 
variable nodes in the graphical structure.
A general introduction to TLDPC codes and a presentation of the new low-rate ensemble are given in Section \ref{sec:tldpc}.
Numerical results are shown in Section \ref{sec:sim}.
 Section \ref{sec:discussion} contains some discussion on the topic.
 
 \section{Necessary Condition for Linear Minimum Distance}
 \label{sec:graph}

\subsection{Common Representation for Sparse-Graph Codes}
\label{sec:framework}
For the sake of generality, we use the following general representation for all sparse-graph codes  
\cite{Tilli04}: 
\begin{definition}[Common construction and base code] 
The construction produces a binary code of length $n$ with the
help of two ingredients: 
\begin{itemize}
\item[(i)] a binary code $\color{blue} \BC$ of rate $\color{blue} \RBC$ of length $m$, with $m \geq n$.
This code is called the {\em {\color{blue} base code}};
\item[(ii)]
a bipartite graph between two sets $V$ and $W$ of vertices of size $n$ and $m$ respectively,  where the degree of any vertex in $W$ is $1$ and the degree of the vertices in $V$ is specified 
by a degree distribution ${\color{blue} \Lambda}=(\lambda_1,\lambda_2,\ldots,\lambda_{s})$ where $\lambda_i$ denotes the
fraction of edges,
incident to vertices of $V$ of degree $i$. 
\end{itemize}
\end{definition}

The bipartite graph together with the base code specifies a code of length $n$ as the set of binary assignments of
$V$ such that the induced assignments\footnote{a vertex in $W$ receives the same assignment as the vertex in $V$ it is 
connected to.} of vertices of $W$ belong to $\BC$. It is straightforward 
to check that the rate of the code obtained by this construction is at least equal to the {\color{blue} {\em designed rate $R$}}, 
$$
R \ \eqdef \ 1 - (1-\RBC)\db,
$$
where $\db$ is the average left degree, which is given by
$$
\db \ \eqdef \ \frac{m}{n} = \frac{1}{\sum_i \frac{\lambda_i}{i}}.
$$

It is common to present the degree distribution $\Lambda$ in its polynomial form $\Lambda(x)=\sum_{i=1}^s \lambda_i x^{i-1}.$

Most sparse-graph code constructions 
can be viewed as a particular instance of this
construction:
\begin{example}[LDPC codes]
The LDPC base code is the juxtaposition of parity codes; 
$\Lambda(x)$ is the left degree distribution.
\end{example}

When the bipartite graph 
has some special structure, we say that it is
a \emph{structured} code ensemble. 
\begin{example}[Parallel turbo codes] 
The base code of a parallel turbo ensemble is the juxtaposition of several convolutional codes, the positions of which are divided into two subsets, the first one is formed by the information bits and the
second one by the redundancy bits. 
The sets $V$ and $W$ in the bipartite graph are also divided into two subsets, the subsets (information and redundancy).
A node in $V$ and a node in $W$ can be connected only if they belong to the same subset type
and all redundancy nodes have degree 1.
\end{example}

The standard decoding procedure \cite{RicUrb:book} for sparse-graph codes is the following. At each  iteration, 
 base code decoding is performed in order to get extrinsic messages for bits of $\BC$, then intrinsic messages at the variable node side are calculated.
After some number of iterations, a posteriori messages of code bits are computed.    
The decoding complexity therefore depends on the complexity of the base code decoding, on the degree distribution of variable nodes (the higher 
the node degree
the more complex 
decoding gets)
and on the number of iterations which are needed to be performed (i.e. the decoding convergence 
speed). 

\subsection{Graph of Codewords of Partial Weight 2}
\label{sec:graph2}
A position in the base code $\BC$ is said to have degree $i$ if it is connected to a node of degree $i$ in $V$. 
Notice that here we allow variable nodes to be of degree $\geq 1$ and, therefore, we allow positions of degree 1 in $\BC$.
The location of these positions has a 
crucial impact on the minimum distance of the overall code, which may become constant in the worst case. 
In what follows, this case is supposed to be avoided.  
To study the minimum distance behavior, we 
make two following definitions.
\begin{definition}[Codewords of $\BC$ of partial weight 2]
Codewords of $\BC$ of partial weight 2 are the codewords 
that involve
exactly two non-zero positions of degree $>1$. 
\end{definition}
\begin{definition}[Clusters] 
A cluster is an ensemble of positions of degree $>1$ in $\BC$, so that for any two positions $i$ and $j$ from this ensemble there exists
a codeword of partial degree 2 in $\BC$ containing $i$ and $j$.   
\end{definition}
The simplest example of clusters can be given in the case of LDPC codes.
\begin{example}[Clusters for LDPC codes, $\lambda_1=0$]
Any two positions of the same parity code form the support for one codeword of partial weight 2. Thus, clusters correspond to ensembles of positions belonging to the same parity codes.
\end{example}\\
With this notion of cluster, we can define now the \graph:
\begin{definition}[Graph of codewords of partial weight 2] The \graph \ is a graph ${\color{blue} G} = \left( \tilde V, E \right)$ with vertex set $\tilde V$ and edge set $E$.
$\tilde V$ is equal to  the set of clusters
and there is an edge $e_{ij}$ between two clusters $\tilde v_{i}$ and $\tilde v_{j}$ 
{iff} there exist 
two positions $x_{k}$ and $x_{l}$ of the base code, belonging to the clusters $\tilde v_{i}$ and $\tilde v_{j}$ respectively, which join 
the same degree-2 variable node. 
\end{definition}
\begin{example}[Graph $G$ for LDPC codes]
The graph $G$ for a LDPC code 
contains clusters that correspond to parity checks in the code structure.
Two clusters are connected if their corresponding parity checks are connected through a degree-2 variable node in the bipartite graph of the code. 
\end{example}

\subsection{Cycles in the Graph of Codewords of Weight 2 and Its Average Degree}
\label{sec:cycles} 
It is well known \cite{DUR04} that the first source of low weight codewords when an LDPC code is chosen
at random are cycles in the Tanner graphs containing only variable nodes of degree
$2$. 
Let us show that they are in one-to-one correspondence with cycles in $G$,
which will allow us to state the necessary condition on linear minimum distance ($d_{min}$).
To do it, we need two following definitions.
\begin{definition}[Node weight]
For a node $v$ in $\tilde V$ and two edges $i$ and $j$ connected to it, we define a node weight $w^{v}_{i,j}$ as follows. By the very definition of a cluster and of the graph $G$, these two edges correspond to two positions of degree $2$ and they form  together with a certain number $a$ of positions of degree $1$
the support of a codeword of partial weight $2$. We let $w^{v}_{i,j}$ be equal to 
this number $a$.
\end{definition}

\begin{definition}[Cycle weight]
The weight $l$ of a cycle ${\cal C}=(V_{\cal C},E_{\cal C})$ in $G$ is equal to 
$l=|E_{\cal C}|+\sum_{v \in V_{\cal C}} w^v,$
where $w^v$ is the node weight associated with 
vertex $v$ in $V_{\cal C}$ and 
the 
two edges in $E_{\cal C}$ connected to $v$. 
\end{definition}

Here is a fundamental relation between cycles in $G$ and low-weight codewords of $\BC$:

\begin{proposition} \label{cyc} A cycle of weight $l$ in $G$ induces a codeword of weight $l$ in the sparse-graph code.
\end{proposition}
 
\begin{proof}
If ${\cal C}=(V_{\cal C},E_{\cal C})$ is a cycle in $G$, we associate to it a configuration $\textbf{x} = \left(x_{1},x_{2},\ldots,x_{m}\right)$ of positions of the base code 
in which 
a) positions of the base code of degree 2 are 
set 
to 1 if in the Tanner graph they are connected to the variable nodes of degree 2
that
are associated with edges in $E_{\cal C}$;
b) a set $B$ of positions of degree 1 is set 
to 1 if they form a codeword of the base code 
of partial weight 2 with two corresponding positions of degree 2;
c) all other positions in $\textbf{x}$ are set 
to 0.

Denote by $w^v$ the size of the set $B$ for a node $v \in V_{\cal C}$.
The point is that the configuration $\textbf{x}$ is  obviously a codeword of the
base code. It has  weight $2 |E_{\cal C}|+\sum_{v \in V_{\cal C}}w^v$. 
$2 |E_{\cal C}|$ non-zero bits of $\textbf{x}$ are connected to degree-2 variable nodes and the rest of them is connected to 
degree-1 variable nodes. 
Thus, there are $|E_{\cal C}|+\sum_{v \in V_{\cal C}} w^v$ 
variable nodes participating in the configuration $\textbf{x}$, 
and they correspond to a codeword of weight $|E_{\cal C}|+\sum_{v \in V_{\cal C}} w^v$. 
\end{proof}

Notice that the weight of the smallest cycle in $G$ is an upper bound on $d_{min}$.
Therefore:

{\corollary
\label{cor1}
If all the node weights $w^v_{i,j}$ of a given
graph $G$ are smaller than some constant $a>0$, $a \in {\mathbb N}$,  then
the minimum distance of its corresponding sparse-graph code is upper bounded by $(a+1) |E_{\cal C}|$.
}

{\corollary
\label{cor2}
If  $G$ contains a cycle of logarithmic weight, $d_{min}$ of the code is logarithmic in the $n$. 
}

Corollary \ref{cor2} can be equivalently expressed in terms of the average degree of $G$.

\begin{theorem}[Upper bound on $\dmin$] \label{t1} 
Consider a sparse-graph code for which the corresponding graphs $G$ have node weights 
upper bounded by a small positive integer $a$. If all the average degrees of these graphs is greater
than $2+\epsilon$  for some $\epsilon >0$, then 
$d_{min}$
grows logarithmically in $n$.
\end{theorem}

{\it Proof.} Consider a sparse graph code. Let $G$ be the associated \graph \ and
$\dmin$ be the minimum distance of the code.
Let $g$ be the girth of $G$ and $\dmoy$ be its average degree. 
By Corollary \ref{cor1} we know that $\dmin \leq (a+1)g$.
 To upperbound this last quantity we use the Moore bound 
for irregular graphs \cite{AHL02} which asserts that the number of vertices $n$ of $G$ satisfies
the following inequality
$
n  \ge  2 \frac{\left(\dmoy -1 \right)^t - 1}{\dmoy-2}
$
where $t = \lfloor \frac{g}{2} \rfloor$.
This implies
$
t \le  \log_{\dmoy-1} \left(  \frac{\dmoy -2}{2}n +1 \right). 
$
 We now conclude by
$\dist  \le (a+1)g 
 \leq (a+1)(2t+1) 
 \le (a+1) \left(2 \log_{\dmoy-1} \left(  \frac{\dmoy -2}{2}n +1 \right) +1 \right). \quad \Box
$


\subsection{Necessary Condition}
\label{sec:necessary-condition}
The following necessary condition follows immediately:
\begin{itemize}
 \item[] {\emph{ While constructing a sparse-graph code ensemble with a linear growth of the average minimum distance, 
cycles of sublinear weights 
in the corresponding graph $G$ of codewords of partial weight $2$ must be avoided.
Or, equivalently, the average degree $\dmoy$ of $G$  must 
be smaller than or equal to 2.
}}
\end{itemize}
\begin{example}[Case of LDPC codes]
Consider 
an 
LDPC 
code ensemble.
Let $\lambda_2$ be 
the fraction of its degree-2 variable nodes and 
let
$\rho$ be
the average degree of its check nodes ($\rho=\frac{m}{r}$, where $r$ is the number of check nodes 
and $m$ is the number of edges).
The number of clusters is equal to $r$. 
To satisfy the necessary condition above, 
$G$ should not have more than $r$ edges.
So, there should be at most $r$ variable nodes of degree $2$ in the bipartite graph. There
are $\frac{\lambda_2 m}{2}$ of such nodes. As
$$
\frac{\lambda_2 m}{2} \leq r = \frac{m}{\rho}; \quad  \lambda_2 \rho \leq 2;
$$  
If we want $d_{min}$
therefore the necessary condition becomes $\lambda_2 \rho \leq 2$.
\end{example}

Note that $\dmoy=2$ is the critical case, when $G$ contains one or several cycles of linear length. 
It has been shown in \cite{TZ2006} that for LDPC codes and $\dmoy=2$, the minimum distance is polynomial in $n$. 
For more general families of sparse-graph codes 
this is not true anymore, 
see for instance Section V of \cite{OTA07}. 

Until now we dealt with 
codes with bounded node weights. 
For some codes the node weights are unbounded, e.g. for turbo codes. 
With a little work, our results can be still extended to unbounded weights, and Corollary \ref{cor2} and Theorem \ref{t1} will hold. 
For completeness of the demonstration, we elaborate the bound for parallel turbo codes, which leads to a much shorter proof of the result by Breiling \cite{Breil}.

\begin{theorem}[\cite{Breil}]
\label{turbo}
$d_{min}$
 of parallel turbo codes 
 grows at most logarithmically in $n$.
\end{theorem}
\begin{proof}
For simplicity, assume only two convolutional components, that both component encoders are recursive systematic convolutional encoders of type $(n,1)$ and 
that 
they are equal.
Then, there exists $t$ such that for any information position $i$ in the convolutional code there is a codeword of partial weight 2 with  information support $\{i,i+t\}$ and with redundancy weight $w$. 
Other codewords of partial weight 2 are deduced by addition. They have information support $\{i,i+kt\}$, their redundancy weight is at most $kw$ and they all belong to the same cluster in $G$. 
Therefore, $G$ consists of (at most) $2t$ clusters \footnote{The factor $2$ comes from the fact that there
are two convolutional codes each one coming with its own set of clusters.}, which are connected through $N$ edges, $N$ is the number of information bits in the turbo code.

Note that the node weights of the clusters are unbounded. 
To circumvent this difficulty, we form smaller clusters by partitioning each cluster into subclusters of size 3 of the form $\{i,i+t, i+2t\}$.
We obtain a new \graph, denoted by $G'$, with 
$2N/3+O(1)$ clusters and of degree $3$.  
Moreover, the node weights of $G'$ are bounded by $2w$. 
Therefore, $G'$ has a cycle of size at most $2 \log_2 (2N/3+O(1))$
and of weight at most $2(1+2w)\log_2 (2N/3+O(1))$.
This yields a codeword of weight  $2(1+2w)\log_2 (2N/3+O(1))$ in the turbo code by Proposition \ref{cyc}.
\end{proof}

\section{On the usefulness of designing sparse graph codes with degree one nodes}
\label{sec:degree_one}

It is 
worthwhile to quote \cite{RicUrb:book} here: {\em ``Given the importance of degree-two edges, it is natural to conjecture that 
degree-one edges could bring further benefits''.} 
This statement can be illustrated by the observation that, from one hand, turbo codes require in general less decoding iterations than
LDPC codes 
and tend to outperform LDPC codes at short blocklengths,
and, from the other hand,
they are decoded with a graphical structure having bits of degree $1$, 
absent in the case of LDPC codes. 
Another confirming example 
is given by \cite[Table VIII]{RicUrb06:me},
where a small fraction of bits of degree 1, present in an LDPC ensemble,
allows to obtain a much steeper waterfall region than in the case of conventional LDPC codes.

Obtaining  codes with a steep waterfall region and a moderate number of decoding iterations becomes 
 problematic in the case of  low code rates:
the number of iterations needed to converge increases when $R$ decreases
(and may go up to several hundreds!), and the error rate curves become very flat. 
The main purpose of this section is to investigate these two phenomena (number of decoding iterations and steepness of performance curves) and to present a heuristic explanation for them, 
which would give us an insight on the design of efficient low-rate codes. 
Being a bit ahead,
let us mention that, in order to design a good low-rate ensemble, one should include bits 
of degree $1$ in the code structure\footnote{or ``hidden'' bits of degree $1$ in the case of LDPC codes
decoded in a turbo-like manner, namely LDPC codes for which all parity-checks  involve at least 
two bits of degree 2.}.

The explanation is given with the help of an EXIT chart on the binary erasure channel\footnote{Note that the same kind of explanation can also be given for other channels
by asserting
that the fundamental relation, namely Theorem \ref{th:area}, which holds for the binary erasure channel, 
holds approximately for other channels.} (BEC).
 For the BEC, the EXIT chart predicts accurately the infinite-length behavior of the code ensemble, and 
represents the ``average'' trajectory for finite blocklengths, which is given by horizontal and vertical
steps between two EXIT curves, the curve of variable nodes and the curve of the base code. Iterative decoding is typically successful\footnote{i.e. successful with probability
tending to $1$ as $n \rightarrow \infty$} if and only if the curve of variable nodes is above the
curve of the base code. The area $\dA$ between both curves has a very nice interpretation : it is linked with the
distance to capacity.  It was
observed in \cite{AKtB04} (generalization of the result first proved in \cite{Shokr99}) that, 
in order to get a capacity-achieving sequence of codes in the sense of \cite{Shokr99}, the quantity $
\dA$ in the sequence should go to $0$. 

To present our explanation, let us define the EXIT curves.
\begin{enumerate}
\item the {\em EXIT curve of the variable nodes :} 
When $\lambda_1=0$ and the channel erasure probability is $p$, this curve is given by the
set of points $(p \lambda(x),x)$ where $x$ ranges over $[0,1]$. When $\lambda_1 >0$, this curve is given by the set of points 
$\left\{(\frac{p (\lambda(x)-\lambda_1)}{1-\lambda_i},x),x \in [0,1] \right\}$. If we bring the degree distribution of the edges of 
degree $>1$, 
\begin{equation}
\label{eq:def_lambda_tilde}
{\color{blue} \tilde{\lambda_i}} \eqdef \frac{\lambda_i}{1-\lambda_1}
\end{equation} for $i >1$ (and $\tilde{\lambda_1}=0$) and the associated polynomial
\begin{equation}
\label{eq:def_poly_lambda_tilde}
{\color{blue}\tilde{\Lambda}(x)}=\sum_{i>1} \tilde \lambda_i x^{i-1} = \sum_{i>1} \frac{\lambda_i}{1-\lambda_1 } x^{i-1},
\end{equation}
 then the EXIT chart of variable nodes
is given by the curve $\left\{(p \tilde{\Lambda(x)},x),x \in [0,1] \right\}$.
\item The {\em EXIT chart of the base code} 
is the curve which relates the fraction of erased messages, ingoing to the base code, with the fraction of outgoing erased messages after the base code decoding, under assumption of the infinite base code length.  
In some cases this EXIT curve can be described analytically.
For example, 
for a right-regular LDPC code, 
this curve is given by the set of points $\left\{(x,1-(1-x)^{r-1}),x \in [0,1]\right\}$. 
\end{enumerate}
For the infinite-length case, the iterative decoding converges if and only the base code EXIT curve lies below the EXIT curve of the variable nodes.
The statement we are going to give below is not really stated in \cite{AKtB04}, but is in essence only a corollary of 
the results given in this paper. 

\begin{Theorem}{[Area theorem]}
\label{th:area}
Let $\dA$ be the area between the two EXIT curves. Then
$$ \dA = \frac{C(p) - R}{\db (1 - \lambda_1)},$$ where
$C(p)$ is the capacity of the 
BEC with probability $p$, $C(p) = 1-p$.
\end{Theorem}
The proof of the theorem is given in Appendix.
Note that this result raises several comments: 

\begin{itemize}
\item For the same gap to capacity and fixed $\db$, the area between the EXIT 
curves of 
variable nodes and of the base code is larger in the
presence of degree-1 nodes than without them by a factor of $\frac{1}{1 - \lambda_1}$.
This can be quite significant, if $\lambda_1$ is large.
\item Although the number of  iterations does not necessarily decreases with $\dA$ (because it also depends on the shapes of both EXIT
curves), in many cases it does. As the presence of degree-1 nodes makes two EXIT curves to lie far from each other, it 
helps to decrease the number of iterations. Note that turbo-codes, especially low rate turbo-codes,
have a large $\lambda_1$, 
which may explain
the small amount of iterations needed 
for their convergence,
in comparison to LDPC
codes, decoded by the standard Gallager algorithm and not having
degree-1 nodes at all.
\item Increasing of $\dA$ has also a positive influence on the slope
in the waterfall region as it was put forward in \cite{Lee:phd,LeeBla03,LeeBla07}. 
This might  be the explanation why turbo-codes
are believed to outperform LDPC codes for moderate lengths. In this case, it is essential to have
a steep waterfall region\footnote{
We refer the reader to \cite{AMRU04} for a rigorous derivation of the exponential behavior 
of the error probability, suggested in the aforementioned references, shown for LDPC codes over the BEC. 
The generalization of the result on turbo-like ensembles is given in \cite{And09}.
For a generalization of the formulas from \cite{AMRU04} to more general channels, see \cite{EMU07,EMOU08}.
}.
\end{itemize}
A straightforward corollary of Theorem \ref{th:area} is
\begin{corollary}
\label{pr:derivative_area}
$$
\frac{d \dA}{d p} = \frac{1}{\db(1-\lambda_1)}.
$$
\end{corollary}

To illustrate 
this point, let us consider an example of a particular TLDPC code family, which will be defined in Section \ref{sec:tldpc}. 
It has the rate $R=\frac{1}{10}$ and the fraction $\lambda_1=\frac{1}{3}$. 
This code family is almost capacity-achieving for the BEC, where it corrects up to 
$89.6\%$ channel erasures:  for $p_0=0.896$, two EXIT curves, 
drawn by straight lines in Fig.\ref{fig:EXIT},  touch each other.
Fig.\ref{fig:EXIT} also presents the EXIT curves (dashed lines), obtained for 
$p=p_0- 0.07$. One can see that they lie much further apart, as predicted. 
Now, to estimate qulitatively  the speed of moving of two EXIT curves, let us compare them with the EXIT curves of an LDPC code ensemble of rate $\frac{1}{10}$. 
For this, we choose an LDPC ensemble with check nodes of degrees $2$,$3$ and $4$, the edge connections to which are described by the check degree distribution 
$\rho(x) = \frac{1}{10}x+\frac{1}{2}x^2+\frac{2}{5}x^3$ (see \cite{RicUrb:book} for definition of $\rho(x)$).
Such a choice of $\rho(x)$ makes the shapes of EXIT curves for the TLDPC base code and for the LDPC base code similar to each other, which allows to have a fair comparison.
To design an LDPC code with parameters similar to those of the TLDPC code, i.e. of rate close to $\frac{1}{10}$ and with maximum variable node degree $12$,
we choose $\Lambda(x)$ to be 
$
\Lambda(x) = 0.486x+0.165x^2+0.037x^3+0.15x^4+0.132x^{10}+0.03x^{11}.
$
The 
ensemble has
the rate
$R \approx \frac{1}{10}$ and 
the threshold
$p_0 \approx 0.8933$. 
Fig.\ref{fig:EXIT_LDPC} shows its EXIT curves 
at the threshold $p_0$ and for $p=p_0-0.07$.
At $p=p_0-0.07$, the EXIT curves of the base code and of variable nodes are much closer
than they are in the TLDPC case, 
as the EXIT curve of the base code does not change
with $p$: it is always given by 
the function $x \mapsto 1-\rho(1-x)$. 

The situation becomes different in the presence of bits of degree $1$, as in the TLDPC case. When the channel improves, the EXIT curve of the base code moves below of its initial position, obtained at the threshold $p_0$. The gain in the area is quantified by Proposition \ref{pr:base_code}, and the area
$\dA_1$ between the EXIT charts of the base code at $p_0$ and at $p_0-\Delta p$ is given by
$$
\dA_1 = \frac{\lambda_1}{1-\lambda_1} \Delta p.
$$
As an example, Fig.\ref{fig:EXIT_2} shows the area for given TLDPC code of rate $\frac{1}{10}$.
This really accounts for the difference between the TLDPC case and the LDPC case and 
clearly results  in a smaller number of decoding iterations, needed to converge.
Moreover, the fact that the EXIT curves of the base code and of variable nodes lie further apart, is very likely to improve the slope in the waterfall region. 

Although the formula given in Proposition \ref{pr:variable_nodes} seems to depend  
on $\lambda_1$ too, this quantity has {\em no} influence at all on how fast the variable node curve moves away with the decrease of $p$. Indeed, the area $\dA_2$ between the EXIT curves of variable nodes at $p_0$ and at $p_0-\Delta p$ is given by
\begin{eqnarray*}
\dA_2  =  \frac{\frac{1}{\db}-\lambda_1}{1-\lambda_1} \Delta p
 =  \frac{\Delta p}{\tilde{\db}},
\end{eqnarray*}
where
$\tilde{\db} \eqdef \frac{1}{\sum_{i>1} \frac{\tilde{\lambda}_i}{i}}$ and the $\tilde{\lambda_i}$'s form the 
degree distribution of variable nodes of degrees $>1$, as defined by (\ref{eq:def_lambda_tilde}).
This is a consequence of the fact that the EXIT chart of variable nodes actually depends on  $\tilde{\Lambda(x)}$ (see 
(\ref{eq:def_poly_lambda_tilde})), and not
on $\Lambda(x)$. 
Such a dependency on $\tilde{\db}$ seems to suggest that, in order to 
improve the performance,
one should try to design sparse-graph codes with $\tilde{\db}$ as small as possible. Ideally, one should get $\tilde{\db}=2$, 
which, by the way, is precisely the case for parallel turbo-codes. 
This consideration provides a heuristic explanation for the common belief that sparse graph codes with a small $\tilde{\db}$
give a good iterative decoding behavior for small and moderate lengths (i.e. the slope of the
waterfall region).

Also note that the case $\tilde{\db}=2$ corresponds to $\tilde{\Lambda}(x)=x$ and, therefore, the EXIT curve of variable nodes is then the straight line $x=py$. Hence, an almost capacity-achieving ensemble in this case should be designed on a base code,  the EXIT curve of which is close to $x=py$. 
We succeeded to obtain this behavior for base code curves of the TLDPC code family, defined in the following section. 

\section{TLDPC Ensemble of Rate $1/10$ Satisfying the Necessary Condition on $d_{min}$
}
\label{sec:tldpc} 
TLDPC codes is a structured code family, first proposed in \cite{ATC05} to meet the requirements of a low iterative decoding complexity, of linear $d_{min}$ and of iterative threshold close to the channel capacity. They can be viewed as a slight modification of LDPC codes
which allow to have degree-$1$ variable nodes by adding some state nodes to the graph structure. They differ from the multi-edge approach suggested in \cite{RicUrb06:me}
in two points: (i) the TLDPC base code is not a juxtaposition of single parity-check codes but it is a tail-biting convolutional code
with 
binary state nodes, (ii) its structure permits a one-dimensional optimization of $\tilde \lambda(x)$, and not a multi-dimensional optimization
as is the case of multi-edge LDPC codes.
They have been designed by using several construction methods, combined together; some of the methods apply to the base code, and some 
concern the bipartite graph.

\subsection{Definition of TLDPC Codes}
\subsubsection{Definition}


For the moment, suppose that 
$\lambda_1=0$.
Then the TLDPC base code is defined as follows:
\begin{definition}[TLDPC base code] 
The base code $\BC$ of the TLDPC code is a tail-biting convolutional code, the Tanner graph of which is presented in Fig.\ref{fig:base}.
The $\bullet$'s are associated with positions of the base code, white vertices  with non-transmitted 
states, and the $\oplus$'s represent parity-check equations. 
The first and the last state nodes are identified. 
The number of $\bullet$'s 
associated to the $i$-th $\oplus$
is denoted by $\color{blue} b_i$.
\end{definition}

In the presence of degree-1 bits ($\lambda_1 >0$), the TLDPC base code is defined in a similar matter, yet the positions of degree 1 in $\BC$ 
have to be specified. 
It is to mention that 
systematic RA (Repeat-Accumulate) codes, systematic IRA codes (Irregular Repeat-Accumulate) codes 
and most of the LDPC codes which are standardized\footnote{i.e. those LDPC codes which have the same amount of degree $2$ variable nodes
as there are parity-checks and where these parity-check nodes are connected together by  a single chain of degree $2$ variable nodes.} are in fact a subclass of TLDPC codes, once 
they are decoded as a turbo-code and not as an LDPC code. 
All these codes have particular TLDPC base codes (see Fig. \ref{fig:IRA}), for which
all $b_i$'s are equal to $1$ for even values of $i$ 
and where the corresponding variable nodes are all
chosen to be of degree $1$. The positions of degree $1$ are redundancy bits of the code.
 
The important feature of the EXIT curve for the defined TLDPC base code is that it is close
to 
a straight line (see previous section for the discussion on it). 
Moreover, the base code is not more complex to decode than single parity-check codes, 
and clearly is much easier to decode that the convolutional code in the underlying structure of 
of turbo codes. 
The TLDPC base code also allows to have a larger $\lambda_2$ under condition of linear $d_{min}$, when compared with conventional
LDPC codes, 
which is helpful for the speed of iterative decoding convergence and
for the waterfall region.
%

In order to
design 
code ensembles with linear $d_{min}$, 
one more constraint is to be put on the choice of the base code $\BC$, to satisfy the necessary condition given in Section \ref{sec:necessary-condition}: 
{\emph{the clusters, formed by codewords of partial weight 2 in the designed base code, must
have bounded weights}}.
This condition ensures that the graph $G$ 
has a {\em linear} number of clusters, and, hence, a non-zero fraction $\lambda_2$ may be allowed, 
with the condition on linear $d_{min}$ growth still satisfied. 
For an example, note that the condition is not verified for systematic IRA codes, having only one single cluster.
 
\subsubsection{Structure of the bipartite graph}
A constraint on the permutation of edges, connected to degree-2 variable nodes in the bipartite graph, comes from the necessary condition on linear $d_{min}$.
The permutation for edges connected to other variable nodes is generated randomly.

The design of the code ensemble starts with the choice of the base code.
Then, the optimization of the variable node degree distribution is performed, by fitting EXIT curves of variable nodes and of the base code, for a target code rate.
As before, let the degree distribution, renormalized over the degrees $>1$, be
denoted by $\tilde \Lambda(x)$.
Let $d_{cluster}$ be 
the average degree of clusters.
Then, during the optimization, the renormalized fraction $\tilde \lambda_2$ of edges connected to degree-2 variable nodes is required to be smaller than $2/d_{cluster}$, so that the average degree of $G$ 
is smaller than $2$.
Suppose that $\tilde \lambda_2 <2/d_{cluster}$.
At this moment some structure on
$G$ 
is to be chosen, so that $G$ does not contain cycles. 
It seems that the simplest way 
would be to make $G$ to be a union of disjoint paths.
But, in this case, the prediction of the iterative threshold, given by the EXIT curve fitting, is not accurate because of the following reason: 
the EXIT method implicitly assumes that the positions of degree $2$ in $\BC$ are chosen independently 
of each other with probability $\tilde \lambda_2$. 
So, the expected fraction
of clusters of degree
$i$ in $G$ should be 
${s \choose i} \tilde \lambda_2^i (1 - \tilde \lambda_2)^{s-i}$, if all clusters are of size $s$. 
To keep the prediction of the EXIT method accurate, degree-$2$ variable nodes are to be chosen such that the fraction of clusters of degree
$i$ is equal to the expected number. 
It is also needed to choose their positions in order to avoid cycles of sublinear length in $G$.

\subsection{Design of a Low-Rate Ensemble}
\label{sec:low-rate}
The design criteria, proposed above, were previously used in the design of TLDPC codes of rates $1/3$ and $1/2$ in \cite{ATC05, ATC06}, and gave very good results. 
The obtained iterative thresholds are within $0.2-0.5$ dB from the Gaussian channel capacity. 
Moreover, it has been proved that one of the code ensembles has 
$d_{min}$, growing linearly in the blocklength. 
In this paper, we design a 
TLDPC ensemble of rate $1/10$, following the same construction methods. 
For our ensemble, it is possible allow a large non-zero fraction $\lambda_1$ and still to satisfy the necessary condition on linear $d_{min}$. 
In what follows, a low-rate TLDPC base code and a permutation structure for degree-2 variable nodes are suggested.

\subsubsection{TLDPC base code of rate $1/2$}
With the aim of designing 
codes of low rates, 
we propose a TLDPC base code of rate $1/2$, defined by the Tanner graph shown in Fig.\ref{fig:lowrate}. 
Note that here $b_i=1$ for any $i$.
Each third section of the base code is chosen to be of degree 1, i.e. this position is connected to a degree-1 variable node in the bipartite graph. 
Positions of degree 1 are marked in blue
in the figure. 
All other positions have degrees $>1$. 
Such a base code gives rise 
to
a code ensemble with 
$\lambda_1=\frac{1}{3}.$

As for the clusters in the graph $G$, 
they
 correspond to the pattern in the Tanner graph of the base code represented in Fig. \ref{fig:cluster}:
any two positions of degree $>1$ in it give rise to a codeword of partial weight 2.
The cluster degree equals to $4$, and 
$G$ contains as many clusters as there are such subgraphs in the Tanner graphs of the base code. 
To satisfy the necessary condition on the linear $d_{min}$ 
, $\tilde \lambda_2$ should verify
$\tilde \lambda_2 \leq  \frac{1}{2}.$

\subsubsection{Degree optimization over the Gaussian channel and permutation structure for rate $1/10$}
Let us fix the design code rate equal to $1/10$. 
We choose $\tilde \lambda_2$ to be slightly less than
$\frac{1}{2}$, namely $\tilde \lambda_2=0.4$, in order to simplify the structure of $G$. 
First, let us compute the cluster degree distribution $A=(a_0,a_1,a_2,a_3,a_4)$, where $a_i$ represents the fraction of
clusters of degree $i$ in $G$.
If the degree of clusters in $G$ are chosen at random given $\tilde \lambda_2=0.4$, the expected values of the the $a_i$'s would be 
the following figures:
$$a_0=\frac{81}{625}; \ \ a_1=\frac{216}{625}; \ \ a_2=\frac{216}{625}; \ \ a_3=\frac{96}{625}; \ \ a_4=\frac{16}{625}.$$
We choose the $a_i$'s to be equal to these fractions for the reasons explained before.

Let us find a structure of $G$ with this degree distribution, so that $G$ does not contain cycles. 
We choose it to contain the following components which we call ``stars'', ``twigs'' and ``chains'' (see Fig. \ref{fig:structure}).
Namely, we divide the Tanner graph of the base code into subgraphs similar to the one represented in Fig.\ref{fig:cluster} and associate a cluster to each of them.
We assume that the number of clusters $M$ is
divisible by $625$.
The generation of the bipartite graph is then performed by associating clusters in order to form the aforementioned components.
It is straightforward to check that this is indeed possible. We summarize in Table \ref{tab:arrangement} the fraction of clusters consumed by each
component. Note that, in Table \ref{tab:arrangement}, an entry for a given component $c$ and a given degree $i$ of the cluster corresponds to the fraction of clusters consumed  in component $c$ which are of degree $i$. 
Using the table, the following three points are easy to check:
1) All clusters are consumed in the components, because the sum of the entries of the column
corresponding to any degree $i$ gives $a_i$.
2) Each entry is nonnegative.
3) The "chains" are possible to form, as the number of clusters of degree $1$, used to form "chains", is even. 
After the degree optimization for the Gaussian channel, 
the following degree distribution was obtained:
$ \tilde \Lambda(x) =0.4x+0.264209x^2+0.090866x^4+0.236716x^8+0.008209x^9. \label{eq:lambda}
$
\section{Numerical Results}
\label{sec:sim}
Let us present performances of TLDPC codes of rate $1/10$ and of lengths $6250$, $18750$, $50000$ and $62500$ over the Gaussian channel.
In each of these cases, $\tilde \Lambda(x)$ of (\ref{eq:lambda}) was adapted to the given blocklength.
The corresponding word and bit error rates, obtained by simulations, are given in Fig.\ref{perfF}. 
The maximum iteration number was fixed to $200$.
It can bee seen in the figure that the estimated decoding threshold is about $-0.8$ dB, which corresponds to the value, obtained with the EXIT method. 
Notice that the threshold is only 
$0.5$ dB away from channel capacity, equal to $-1.286$ dB. This is quite close for these signal to noise ratios, since the capacity
at $-0.8$ dB is only about $0.111$. 
In addition, numerical results did not catch the error-floor,
which is expected to happen thanks 
to the good $d_{min}$ of designed codes.

As for the convergence, for the largest simulated blocklength (62500) and at signal-to-noise ratio -0.5 dB the decoder only needs 86 iterations in average to converge,  
due to the large fraction $\lambda_1$.
Moreover, as the base code can be represented by a 2-state trellis, where each trellis section carries only one bit, the complexity of one decoding iteration is very low. This results in a total low decoding complexity.

\section{Discussion}
\label{sec:discussion}
In this paper, two objectives were followed. 
The first one was to define a necessary condition to design sparse-graph codes with linear minimum distance in the blocklength. 
Such a condition has been found and is expressed either in terms of cycles or 
in terms of the average degree of the \graph.
The second objective was to design a new low-rate, structured code ensemble with such features as a linear $d_{min}$, a small gap to the channel capacity, a low decoding complexity and also a possibility to apply well-developed techniques (EXIT charts, density evolution) to optimize the degree distribution of variable nodes.
The aforementioned design has been performed in the framework of TLDPC codes, and a TLDPC code ensemble of rate $1/10$ performing well over the 
Gaussian  channel has been proposed. 

The linear minimum distance property for the presented TLDPC ensemble may be proved  
by using standard techniques based on weight distributions, for instance
by computing the growth rate of the average weight distribution in the asymptotic case and to show that its first derivative at the origin is strictly negative.   
We do not present such a proof in the paper, but we conjecture such a behavior.

\section{Acknowledgment}
Part of this work was done when the first author was with France Telecom R\&D.

\bibliographystyle{plain} 
\bibliography{low-rate,code}  

\appendix
\section{Proof of Theorem \ref{th:area}}
\label{app:1}
First recall (\cite{AKtB04}) that the area under the EXIT curve for the variable nodes is given by
\begin{Proposition}
\label{pr:variable_nodes}
$
\Ar = 1 - p \frac{ \frac{1}{\db} - \lambda_1}{1 - \lambda_1} $.
\end{Proposition}

\proof{}
The area below the curve of variable nodes is 
\begin{eqnarray*}
\Ar  =    1 - \int_{0}^1 \frac{p (\lambda(x)-\lambda_1)}{1-\lambda_1} dx  =  1 - \frac{p \sum_{i>1} \frac{\lambda_i}{i} }{1 - \lambda_1} 
     =   1 - p \frac{ \sum_i \frac{\lambda_i}{i} - \lambda_1}{1 - \lambda_1}   =  1 - p \frac{ \frac{1}{\db} - \lambda_1}{1 - \lambda_1}.  \quad \Box
     \end{eqnarray*}

The area below the EXIT chart of the base code is given by 
a corollary of \cite[Theorem 1]{AKtB04}:
\begin{Proposition}
\label{pr:base_code}
Assume that the bits of degree $1$ of the base code $\BC$ can be completed to form an information set for $\BC$. Then
the area ${\mathcal A}$  under the EXIT curve of the base code over the BEC is given by
$
\Ar =  \frac{\RBC - (1-p)\lambda_1}{1 - \lambda_1},
$ 
where $\RBC$ denotes the rate of the base code.
\end{Proposition}
\proof
From Theorem 1 (\cite{AKtB04}) we know that
$
\Ar = \frac{H(V|Y)}{(1 - \lambda_i)m}.
$
Here  $V$ consists of a codeword of the base code which is chosen uniformly at random 
and $Y$ is the transmitted codeword
where all positions of degree $>1$ have been erased and all positions of degree $1$ have been
erased with probability $p$. Let $Z$ be the number of non-erased positions of
$V$. Note that 
$
H(V|Z=t) = \RBC m - t,
$
by the assumption made on the positions of degree $1$.
So,
$
H(V|Y) = \RBC m -(1-p)\lambda_1 m,
$ and the proposition follows immediately. $\Box$

We are ready now for the proof of Theorem \ref{th:area}.

\proofP{of Theorem \ref{th:area}}
As long as the EXIT curve of the base code lies below the EXIT curve of the variable
nodes,
by Propositions \ref{pr:base_code} and \ref{pr:variable_nodes} 
\begin{eqnarray*}
\dA & = & 1 - p \frac{ \frac{1}{\db} - \lambda_1}{1 - \lambda_1}  -
 \frac{\RBC - (1-p)\lambda_1}{1 - \lambda_1} 
  = \frac{\db(1 - \lambda_1)- p(1- \lambda_1 \db) -\RBC \db +(1- p)\lambda_1 \db}{\db (1-\lambda_1)}\\
 & = & \frac{ (1-\RBC) \db-p}{\db(1-\lambda_1)} 
  =  \frac{C(p) - R}{\db (1-\lambda_1)}. \quad \Box
\end{eqnarray*}

\newpage 
\begin{figure}[h!]
\includegraphics[scale=0.5]{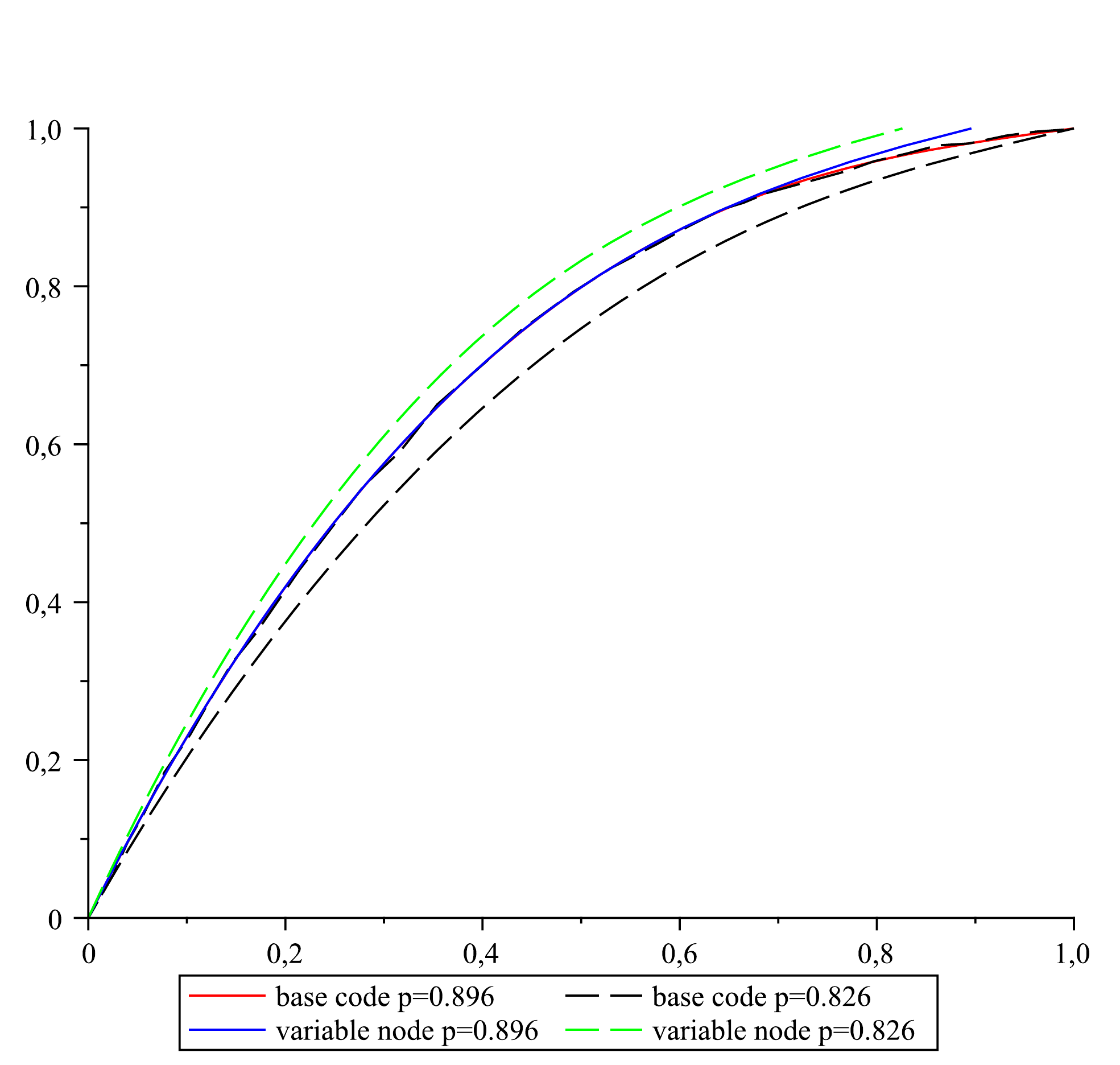}
\caption{ \label{fig:EXIT} EXIT chart of a TLDPC code of $R=\frac{1}{10}$ code with 
$\lambda_1=\frac{1}{3}$} 
\end{figure}

\begin{figure}[h!]
\includegraphics[scale=0.5]{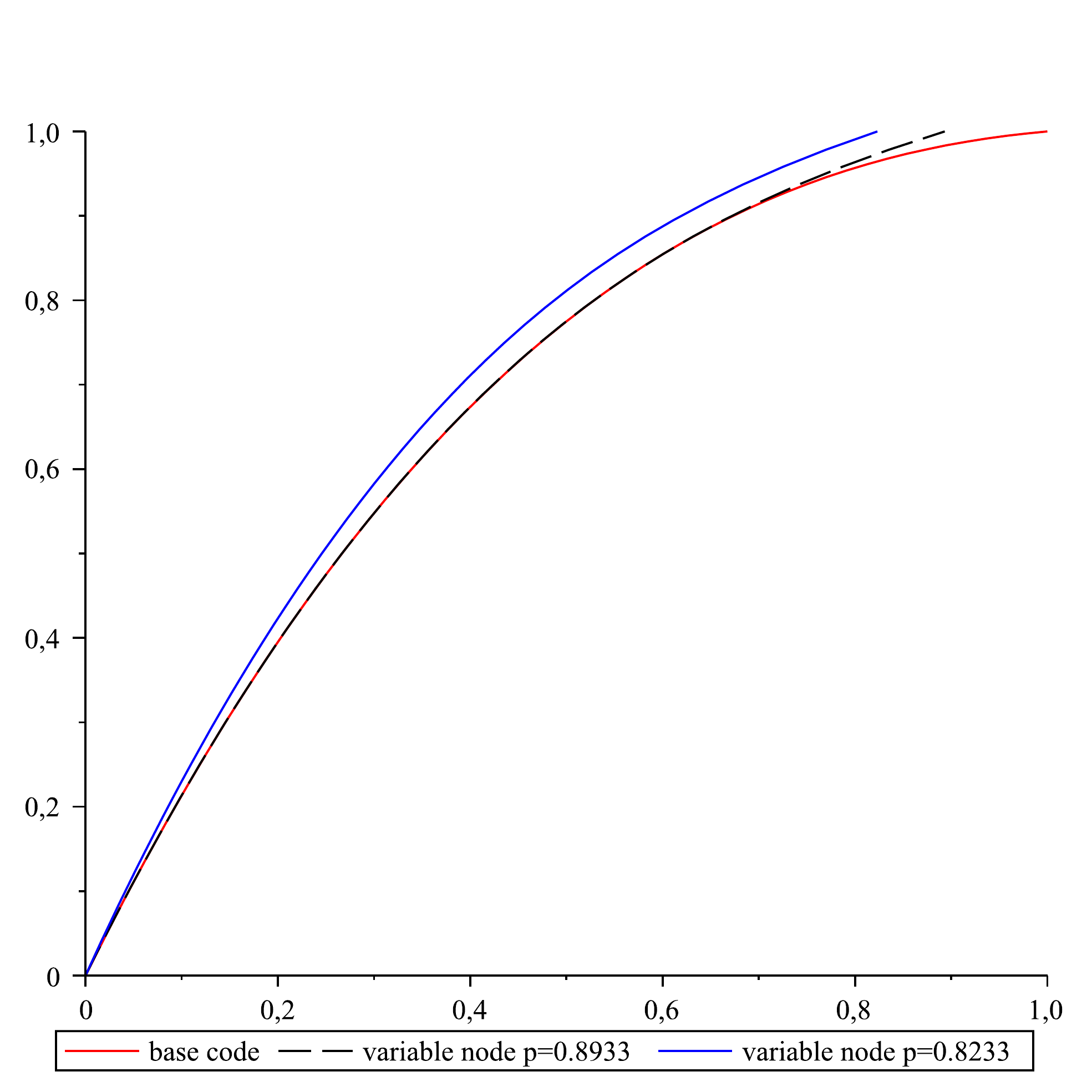}
\caption{ \label{fig:EXIT_LDPC} EXIT chart of an LDPC code of rate $\frac{1}{10}$ code.} 
\end{figure}

\begin{figure}[h!]
\includegraphics[scale=0.5]{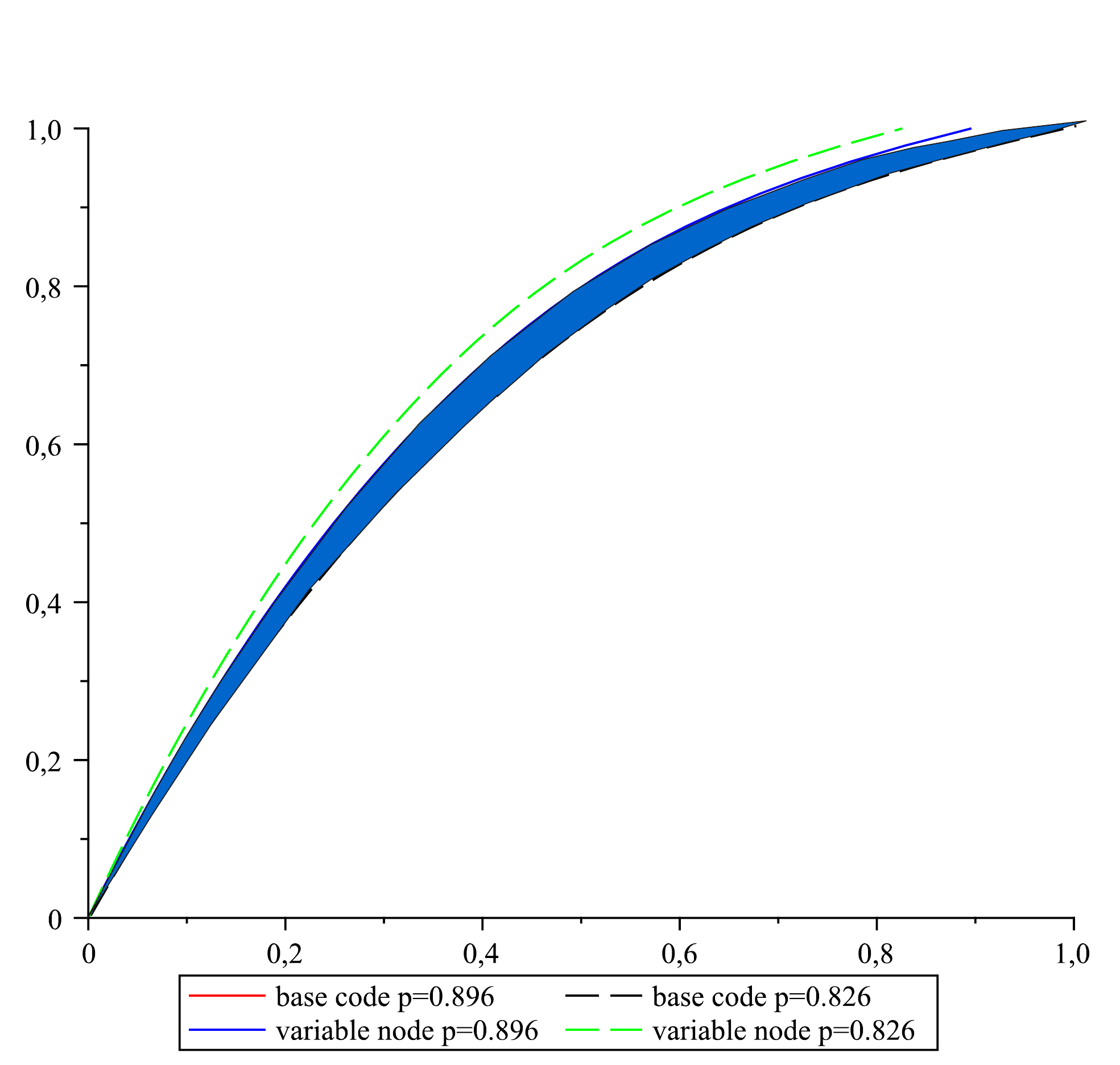}
\caption{ \label{fig:EXIT_2} $\dA_1$ for the TLDPC code of rate $\frac{1}{10}$.} 
\end{figure}

\begin{figure}[h]
\begin{picture}(0, 90)
\put(15,-155){\includegraphics[scale=0.4]{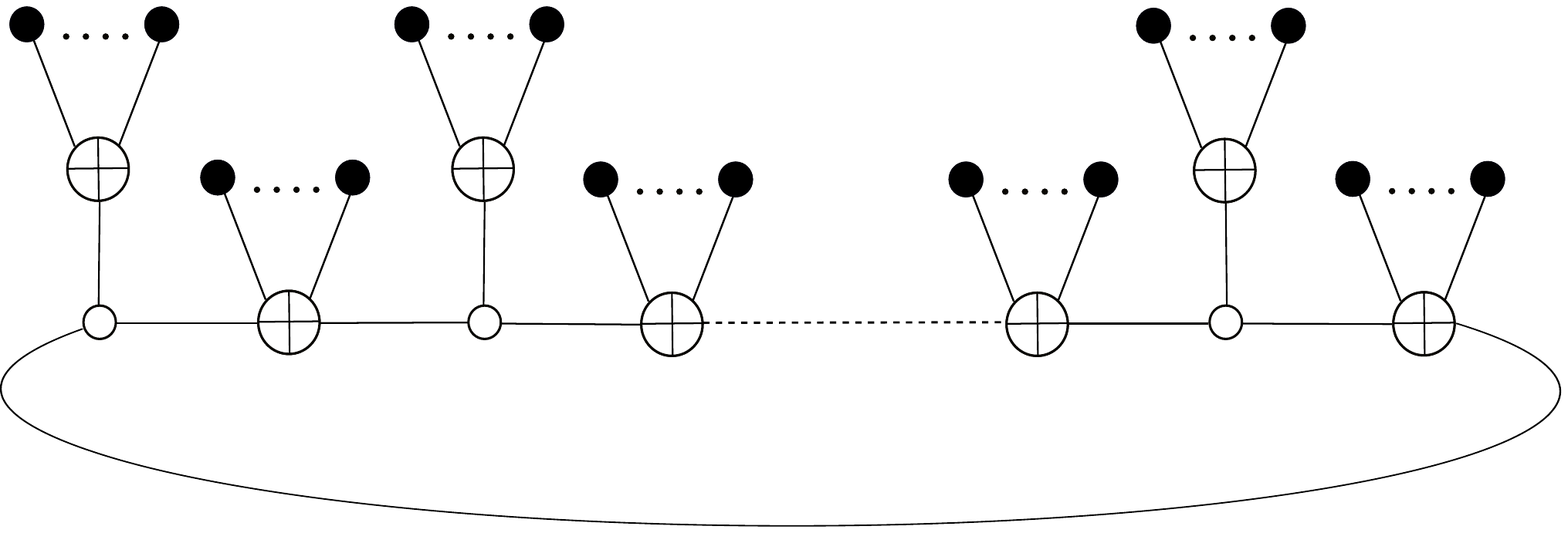}}
\put(23,85){$b_0$} \put(55,60){$b_1$} \put(82,85){$b_2$} \put(110,60){$b_3$}
\put(165,60){$b_{r-3}$} \put(192,85){$b_{r-2}$} \put(220,60){$b_{r-1}$}
\end{picture}
\caption{ \label{fig:base} Tanner graph of a TLDPC base code.} 
\end{figure}

\begin{figure}[h!] 
\epsfig{file=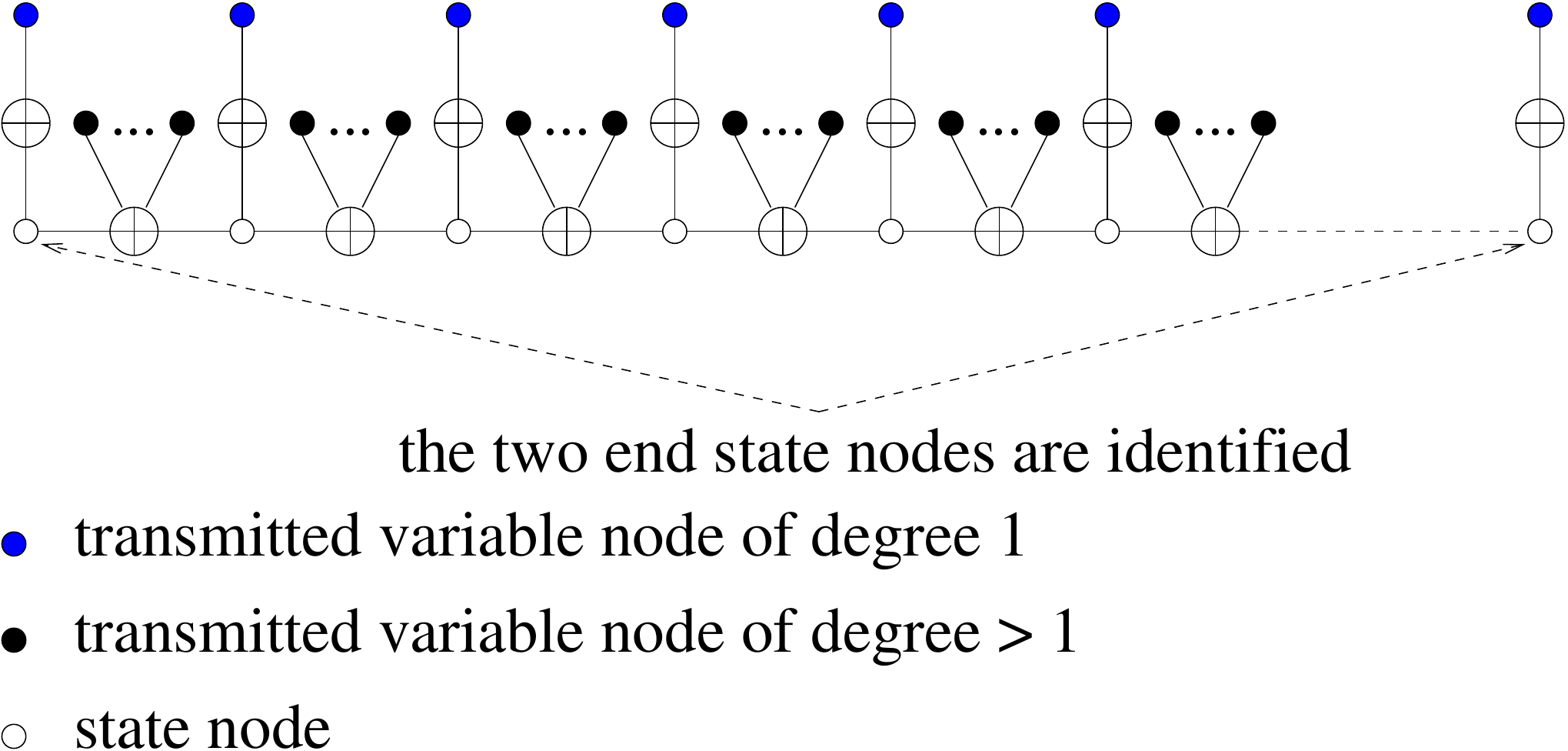,height=4cm} 
\caption{\label{fig:IRA} 
Base code for systematic (I)RA codes and standartized LDPC codes.} 
\end{figure} 

\begin{figure}[h!]
\includegraphics[scale=0.4]{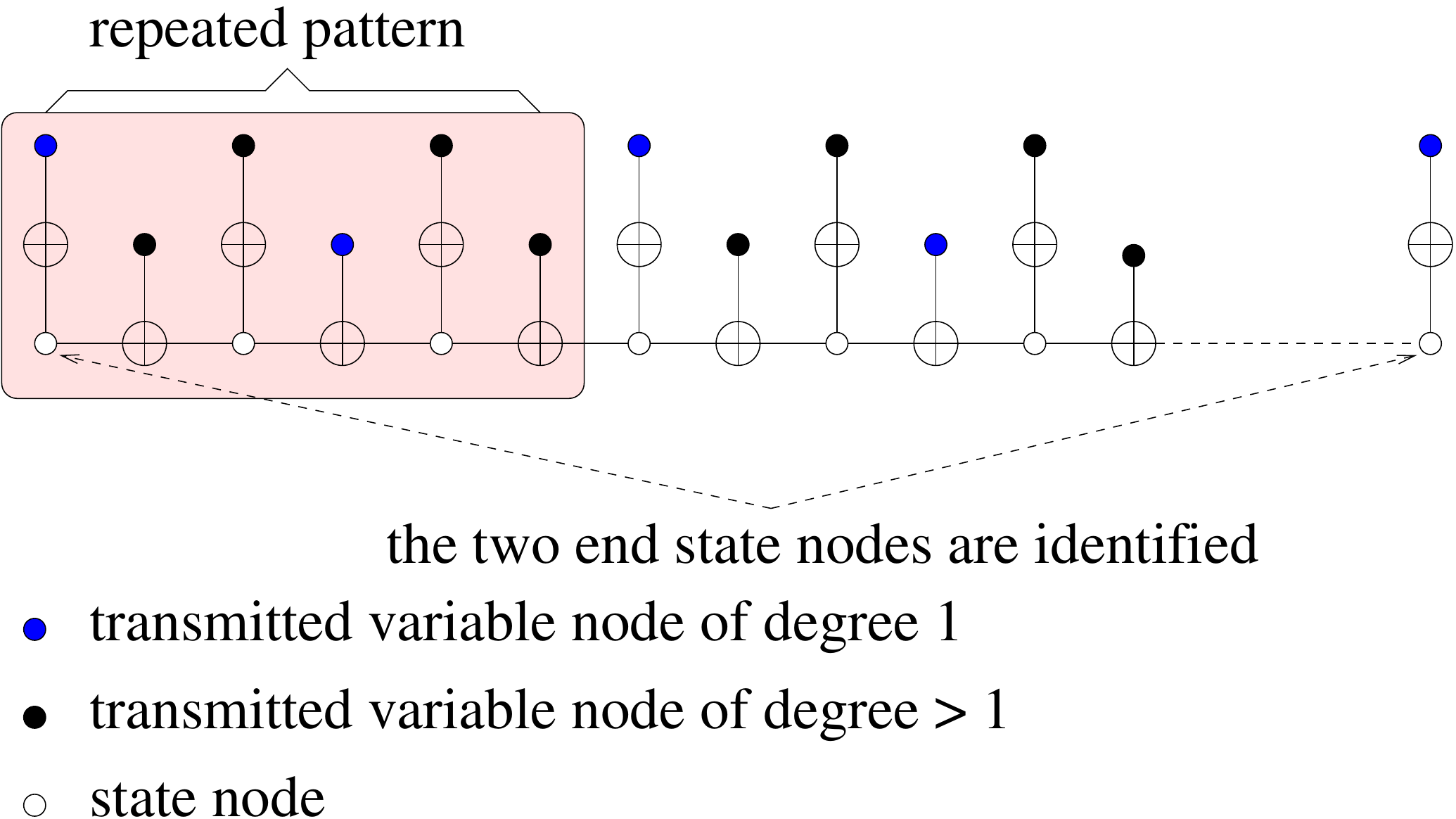}
\caption{ \label{fig:lowrate} Tanner graph of a TLDPC base code of rate $1/2$.} 
\end{figure}

\begin{figure}[h!]
\includegraphics[scale=0.6]{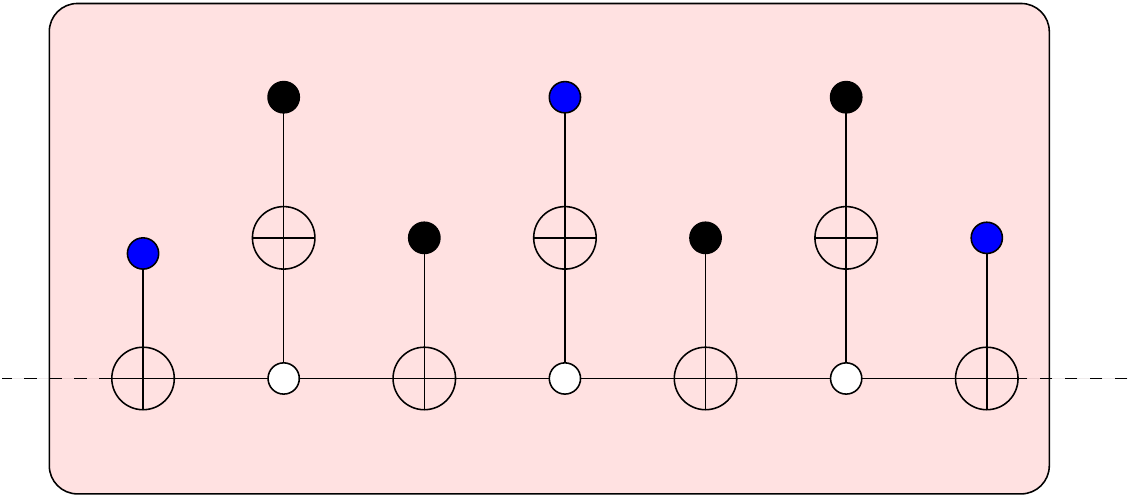}
\caption{ \label{fig:cluster} Pattern in the Tanner graph of the TLDPC base code of rate $1/2$ giving rise to a cluster.} 
\end{figure}

\begin{figure}[h]
\begin{picture}(0, 250)
\put(-20,-120){\includegraphics[scale=0.5]{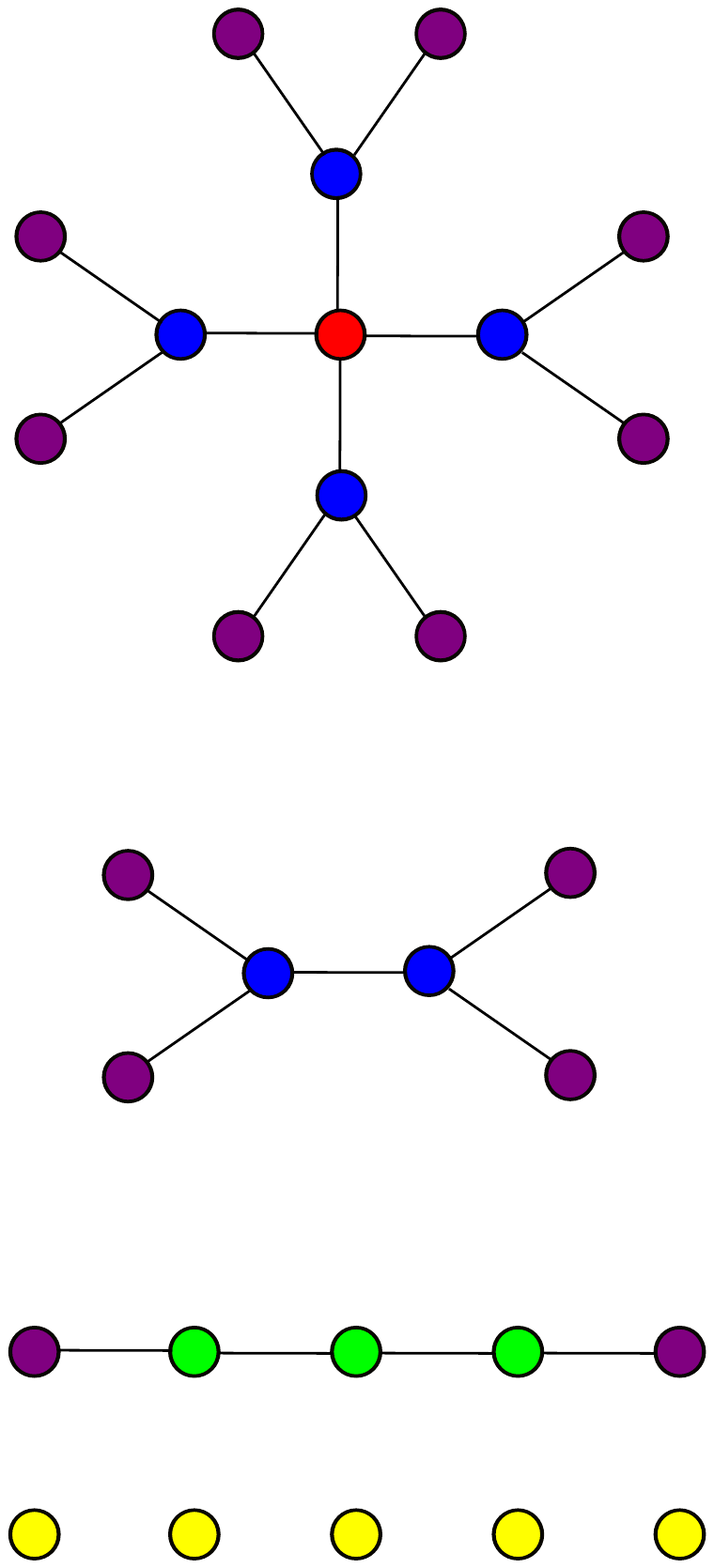}}
\put(190,190){``star''} \put(180,90){``twigs''} \put(185,35){``chain''} \put(170,5){degree-$0$ clusters}
\end{picture}
  \caption{\label{fig:structure} Configurations in the structure of the \graph. Clusters of different degrees have a different color.}
\end{figure}

\begin{figure}[h]
{\includegraphics[scale=1]{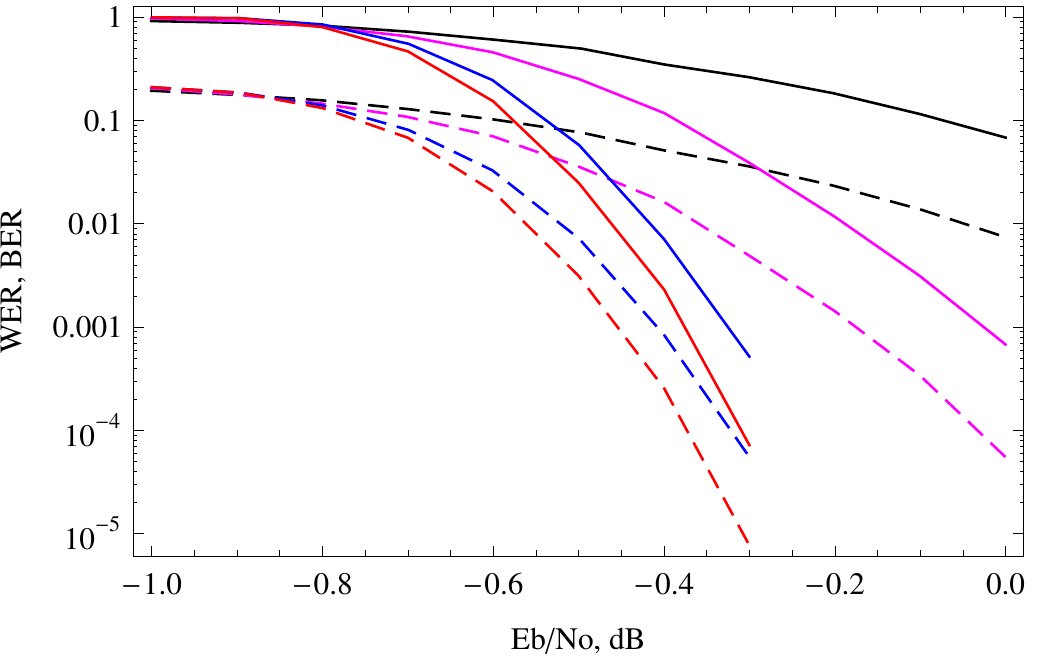}}
\caption{ \label{perfF} Performance of TLDPC codes of lengths (from right to left) 6250, 18750, 50000 and 62500 with $\lambda_1=1/3$ and 
the degree distribution (\ref{eq:lambda}). 
Solid 
lines represent word error rates and dashed lines - binary error rates.
} 
\end{figure}

\begin{table}[h!]
\caption{\label{tab:arrangement} Arranging clusters to form the components. }
\begin{center}
\begin{tabular}{|c|c|c|c|c|c|}
\hline 
 & 0 & 1 & 2 & 3 & 4  \\
 \hline
 ``star'' & 0 & $8a_4$& 0 & $4a_4$ & $a_4$\\
 \hline
  ``twig'' & 0 & $2(a_3-4a_4)$& 0 & $a_3-4a_4$ & $0$\\
 \hline
  ``chain'' & 0 & $a_1-2a_3$& $a_2$ & $0$ & $0$\\
 \hline
 ``isolated cluster'' & $a_0$ & $0$& 0 & $0$ & $0$\\
 \hline
\end{tabular}
\end{center}
\label{default}
\end{table}

\end{document}